\documentclass[acmsmall]{acmart}
\usepackage{etc}
\usepackage{makecell}
\usepackage{float}
\usepackage{algorithm}
\usepackage{algpseudocode}
\usepackage{graphicx}
\usepackage{multirow}
\usepackage{textcomp}
\usepackage{hyperref}
\usepackage{cleveref}
\usepackage{arydshln}
\usepackage{booktabs}
\usepackage{subcaption}
\usepackage[export]{adjustbox}
\usepackage{balance}
\usepackage{wrapfig}
\newtheorem{lemma}{Lemma}
\crefname{lemma}{Lemma}{Lemmas}

\AtBeginDocument{%
  }

\setcopyright{cc}
\setcctype{by}
\acmDOI{10.1145/3715773}
\acmYear{2025}
\acmJournal{PACMSE}
\acmVolume{2}
\acmNumber{FSE}
\acmArticle{FSE054}
\acmMonth{7}
\received{2024-09-05}
\received[accepted]{2025-01-14}

\begin{document}

% \title{\textit{The Less I know the \textcolor[rgb]{0,0.5,0}{Greener}}:\\ An adaptive language-agnostic pruning method for language models for code}
\title{	An Adaptive Language-Agnostic Pruning Method for Greener Language Models for Code}
% Some alternatives
% 1. Ignorance is Green: An adaptive ...
% 2. Ignorance is evergreen: An adaptive ...
% 3. An adaptive language-agnostic pruning method for greener language models for code

\author{Mootez Saad}
\orcid{0009-0008-8159-3632}
\affiliation{%
  \institution{Dalhouise University}
  \city{Halifax}
  \country{Canada}}
\email{mootez@dal.ca}

\author{Jos\'e Antonio Hern\'andez L\'opez}
\orcid{0000-0003-2439-2136}
\affiliation{%
  \institution{Linköping University}
  \city{Link\"oping}
  \country{Sweden}}
\email{jose.antonio.hernandez.lopez@liu.se}

\author{Boqi Chen}
\orcid{0000-0002-1451-3603}
\affiliation{%
  \institution{McGill University}
  \city{Montr\'eal}
  \country{Canada}}
\email{boqi.chen@mail.mcgill.ca}

\author{D\'aniel Varr\'o}
\orcid{0000-0002-8790-252X}
\affiliation{%
  \institution{Link\"oping University}
  \city{Link\"oping}
  \country{Sweden}}
\email{daniel.varro@liu.se}

\author{Tushar Sharma}
\orcid{0000-0002-0538-052X}
\affiliation{%
  \institution{Dalhouise University}
  \city{Halifax}
  \country{Canada}}
\email{tushar@dal.ca}


\begin{abstract}
% \todo{JA. I think that it is a little bit long.}
Language models of code have demonstrated remarkable performance across various software engineering and source code analysis tasks. However, their demanding computational resource requirements and consequential environmental footprint remain as significant challenges.
This work introduces \alpine{}, an {a}daptive programming {l}anguage-agnostic {p}run{in}g techniqu{e} designed to substantially reduce the computational overhead of these models. 
The proposed method offers a pluggable layer that can be integrated with all Transformer-based models.
With \alpine{},
input sequences undergo adaptive compression throughout the pipeline, reaching a size that is up to $\times 3$ less their initial size, resulting in significantly reduced computational load.
Our experiments on two software engineering tasks, \textit{defect prediction} and \textit{code clone detection} across three language models \cb{}, \gcb{} and \unx{} show that \alpine{} achieves up to a 50\% reduction in FLOPs, a 58.1\% decrease in memory footprint, and a 28.1\% improvement in throughput on average. This led to a reduction in CO\textsubscript{2} emissions by up to $44.85$\%. Importantly, it achieves a reduction in computation resources while maintaining up to 98.1\% of the original predictive performance. 
These findings highlight the potential of \alpine{} in making language models of code more resource-efficient and accessible while preserving their performance,
contributing to the overall sustainability of their adoption in software development.
Also, it sheds light on redundant and noisy information in source code analysis corpora, as shown by the substantial sequence compression achieved by \alpine{}.

\end{abstract}

\maketitle

\section{Introduction}
% \begin{itemize}
%     \item Start by stating the emergence of DL methods and their impact on software engineering. Specifically, language models trained for programming language reasoning.
%     \item Mention the different tasks on which these models were applied.
%     \item State the problem related to the size of these models, and at the same time the reason why such models are big. Here you can introduce the scaling laws by Kaplan \etal{}~\cite{Kaplan2020} that state that model performance is proportional to the model and data size. However, on the other hand, this redeems such models expensive to finetune and run. Can also include the environmental impact.
%     \item This leads to the need to find methods to make these models more computationally efficient without compromising too much performance. We can do this if we look at the problem from a data standpoint. The hypothesis is that not all data is equal. Some can be discarded, which entails a reduction in computation.
%     \item Introduce the pruning method. State how first, it is simple, second effective (\textit{insert numbers}) and also extensible because it is based on the Transformer architecture which is now ubiquitous in many DL methods. This point will also be divided into bullet points since it showcases the contribution of this work.
%     \item Conclude by placing the replication package.
    
% \end{itemize}
% \assignedto{Tushar - done with the first pass.}

Transformer-based language models~\cite{Vaswani2017} were initially conceived for natural language understanding. Under the \textit{naturalness} hypothesis~\cite{Hindle2012}, they were applied to solve diverse software engineering ({\sc se}) tasks with remarkable accuracy. These tasks include
% span over subfields~\cite{Hou2023LargeLM,Sharma2024} such as 
code generation~\cite{Lai2023ds, Zeng2022Generation, Dibia2023, Chen2022LR}, quality assurance~\cite{Alqarni2022, Ciborowska2023, Fatima2023, Tang2023CSGVD}, maintenance~\cite{Tian2023Best, Zhang2024APTC, LeCong2023, Paul2023automated}, and requirements engineering~\cite{Moharil2023, Ezzini2022}.

Though the language models have shown exceptional promise,
fine-tuning or running inference on these models require significant computational resources.
Moreover, the demand of the resources is continuously increasing with the size of the language models.
% , and this demand increases as models grow larger over time. 
% Fine-tuning or running inference on these models requires significant computational resources, a demand that is amplified as they continue to increase in size over time.  
This entails that their training grows correspondingly, leading to increased usage of cloud computing services which operate on a pay-per-use basis. 
This inflates the direct financial investment for researchers and average users with consumer-grade hardware. 
In addition to the financial impact, the carbon footprint of training and maintaining these models is a growing environmental concern as they contribute to higher CO\textsubscript{2} emissions~\cite{Luccioni2023power, Lacoste2019Quantifying, Wang2013Energy}. 
These issues underscore the urgency for more resource efficient algorithms.

% Shi \etal{}~\cite{Shi2023} have proposed \textit{Compressor} as a step towards this direction. \textit{Compressor}, combines task-specific knowledge distillation~\cite{Hinton2015DistillingTK} and evolutionary search of network architecture. It first fine-tunes a pre-trained language model on a downstream task, then performs a search over the space of hyperparameters of that model's architecture to produce a smaller version of it. The smaller student model is then pre-trained using knowledge distillation through the parent's output and finally fine-tuned on the downstream task. Knowledge distillation often falls short in terms of task-specificity and computational overhead. The distilled student models are typically optimized for a specific task or domain, limiting their generalization abilities and adaptability to new tasks or domains not encountered during the distillation process. This task-specificity can be a significant drawback, especially in scenarios where the model needs to handle diverse tasks or adapt to new domains. Moreover, the process of knowledge distillation itself can be computationally expensive, as it involves fine-tuning the teacher model either way, generating predictions or soft labels, pre-training, and fine-tuning the student model.

Shi \etal{}~\cite{Shi2023} have proposed \textit{Compressor} as a step towards this direction. \textit{Compressor} combines task-specific knowledge distillation~\cite{Hinton2015DistillingTK} and evolutionary search of network architecture to generate a configuration that yields a smaller model. Task-specific knowledge distillation often falls short in terms of task-specificity and computational overhead. The distilled student models are typically optimized for a specific task or domain, limiting their generalization abilities. Moreover, the process of task-specific knowledge distillation itself can be computationally expensive, as it involves \emph{fine-tuning the teacher model either way}, generating predictions or soft labels, pre-training, and fine-tuning the student model.

% To address these challenges, an alternative direction is to explore solutions that improve the efficiency of fine-tuning and inference in language models of code without compromising the model's parameters. We can approach the problem from a \emph{data standpoint}. 

% \textit{DietCode}~\cite{Zhang2022DietCode} does so by relying on the fact that different tokens and statements of the input have different importance. It assigns importance scores to tokens and statements based on the attention weights after performing a forward pass. It formulates the problem as a 0-1 knapsack problem, where statements can be regarded as the items to be collected into a knapsack, with the
% attention weights being the values, the statement lengths being the weights and the constraint being the length of the new input. \textit{DietCode} has two main limitations. First, it relies on the peculiarities of each programming language, such as specific syntax, semantics, and constructs. This reliance makes extending the method to all programming languages difficult without overhead and customization. Second, it assumes that the attention weights assigned during the forward pass of the pretrained model remain static and applicable across different corpora, tasks, and downstream tasks. However, the importance and relevance of tokens and statements can vary depending on the specific context and task. 
Similarly, \textit{DietCode}~\cite{Zhang2022DietCode} attempts to achieve better resource utilization by assigning importance scores to tokens and statements based on attention weights and then selecting the most important statements to form a new, condensed input. However, it has two main limitations. First, it relies on the peculiarities of the task's programming language, such as specific syntax, semantics, and constructs. This reliance makes extending the method to all programming languages difficult without overhead and customization. Second, it assumes that the attention weights remain static and applicable across different corpora and tasks. However, this is not the case given that these attention weights are derived from trainable parameters, which are updated when the language model is fine-tuned on other tasks.

LTP~\cite{Kim2022LTP} covers some of \textit{DietCode}'s limitations. It proposes a method that learns a threshold at each Transformer layer and tokens are pruned if their importance score is below such threshold. Concretely, the pruning is done by zeroing out the hidden representations of these tokens. However, such an approach does not effectively reduce computation. While the method assigns zero values to pruned tokens by multiplying the hidden representations with a mask\footnote{LTP's implementation: \url{https://t.ly/qpu4J}}, the sequence length remains unchanged. Consequently, these zero-valued representations still participate in subsequent matrix multiplications and other operations. In most hardware and software implementations, multiplications involving zero are not automatically optimized out. As a result, the same number of floating-point operations (FLOPs) are performed regardless of how many tokens are ``pruned" leading to no actual reduction in computational cost. Finally, LTP requires multiple fine-tuning steps: in the first pass, it learns the thresholds as well as full model parameters, then fine-tunes the model’s weights in the second pass using those thresholds.

To tackle the aforementioned limitations, we present \alpine{}, a language-agnostic token pruning technique seamlessly integrated into \textit{any} Transformer-based model. This approach involves placing a pruning module within the Transformer block. \alpine{} processes a code snippet, generating a set of importance scores for each token. Based on these scores and a configurable importance range, it completely eliminates tokens that exist outside the range. Consequently, the reduced input size yields considerable speed up {in both fine-tuning and inference phases}. Notably, \alpine{} operates without requiring supplementary code information, such as statement types, ensuring its programming-language independence. Furthermore, the importance scores are computed during the fine-tuning process, enabling it to adapt dynamically to various downstream Software Engineering (SE) tasks. 

We evaluate our proposed method on two software engineering tasks \textit{defect prediction} and \textit{code clone detection} using three language models, \cb{}, \gcb{} and \unx{}. Our approach allows these models to use $\times 2$ less computation, measured by the number of floating point operations (\flops{}) and memory footprint, with minimal loss in performance and even slight improvement in some scenarios. 

This study makes the following contributions to the field.
\begin{itemize}
    \item We conduct an analysis to quantify the computational requirements of the individual components within the Transformer model measured by \flops{}. The analysis helps us identify the bottleneck components and understand the changes in computational complexity with respect to the input sequence length.
    \item Using these insights, we design an effective pruning method that makes language models computationally efficient. Its plug-and-play characteristic makes it easy to integrate and use with virtually any Transformer-based model with minimal overhead.
    \item We illustrate how \alpine{} can maintain a considerable amount of the original accuracy exhibited by the none-pruned models with substantially less computation and carbon emission.
    %\item We demonstrate the impact of \alpine{} from an environmental viewpoint and its contributions to a more sustainable way of using language models in software engineering.
    % MS: This fits better the implication side\item Given its properties of saving on computation while maintaining accuracy, our method can speed up and lower the barriers to the adoption of language models of code in practice, especially with the consumer-grade \gpu{}s.
\end{itemize}

% \todo{added value from SE perspective. MS: Done.)}
\alpine{}'s ability to reduce computational requirements while maintaining accuracy has two significant implications. 
First, 
it contributes to more sustainable software development practices by reducing the required computation and carbon footprint associated with such models.
Second, it lowers the barriers to entry for developers and researchers, facilitating the adoption of language models for code, especially when constrained by consumer-grade \gpu{}s. As a result, it offers a more accessible and environmentally friendly solution for the software engineering community.
\newline
\textbf{Replication package}: Our replication package including source code and data is available online~\cite{replication}.
\section{Background and Motivation}
% \assignedto{Tushar - done with the first pass.}
In this section, we provide a brief primer on the Transformer architecture~\cite{Vaswani2017} and the implications of its computational requirements. Then, we investigate the number of operations performed by the main blocks of this architecture to identify a potential bottleneck within a certain class of Transformer-based models.
% \todo{A brief section outline. You may start by discussing why we are studying transformer architecture, their role in SE tasks (as Daniel mentioned). MS: This was handled at the beginning of the introduction.}
\subsection{Transformer Architecture}
The Transformer architecture, introduced by Vaswani \etal{}~\cite{Vaswani2017} for natural language processing ({\sc nlp}), is a sequence-to-sequence model composed of two main components: an encoder and a decoder. This work focuses on the encoder-only variant, one of the most used variants in SE-related tasks~\cite{Hou2023LargeLM}. We further elaborate in Section~\ref{discussion:extention} how \alpine{} can be extended to encoder-decoder and decoder-only models. For the rest of the paper, we refer to a Transformer encoder layer by a Transformer layer for convenience.

\noindent
\textbf{Input Representation:} The input fed into the Transformer layer consists of the tokenized sequence of tokens $T$ and an attention mask $M$. $T$ includes a special \textbf{\texttt{[CLS]}} token, whose representation is used for classification or ranking tasks, and a \textbf{\texttt{[SEP]}} to indicate the end of the sequence, followed by padding tokens \texttt{\textbf{[PAD]}}, if necessary\footnote{Appending two sequences is also another format that can be used as input: I$= [t_{cls}, t_1, t_2, \ldots, t_{sep}, u_1, u_2, \ldots, t_{pad_m}]$}. 
\begin{equation*}
T = [t_{cls}, t_1, t_2, \ldots, t_n, t_{sep}, t_{pad_1}, \ldots, t_{pad_m}]
\end{equation*}
The attention mask is used to indicate which tokens should be attended to and which ones should be ignored,
\[
M[i] = 
\begin{cases}
    0 & \text{if } t_{i} = \text{PAD}, \\
    1 & \text{otherwise.}
\end{cases}
\]
The input representation is then passed through an embedding layer to obtain dense vector representations for each token, which serve as the input to the encoder layer.

% \todo{Though the above subsection is relevant, it is quite known. We may consider cutting some text from this subsection.}

\noindent
\textbf{Main blocks}: The main blocks of this layer are the Multi-head Self-Attention (\mha{}) which is composed of $h$ attention heads, and the Position-wise Feed-forward network (\ffnn{}). \mha{} is defined as follows
\vspace{-2mm}
\begin{equation}
    \text{\mha{}(x)} = \sum_{i = 1}^{h} Attn_{i}(x), \,\text{y\textsubscript{\mha{}}} = LayerNorm(x + \text{\mha{}(x)}),
    \label{eqn:mha_formulation}
\end{equation}

\noindent where $Attn(\cdot)$ is the scaled dot attention function introduced by Vaswani \etal{}~\cite{Vaswani2017}, and $x$ is the input sequence. The derivation of Equation~\ref{eqn:mha_formulation} is taken from the work of Kobayashi \etal{}~\cite{Kobayashi2020} and Saad and Sharma~\cite{Saad2023}. After the \mha{}, the \ffnn{} block is applied to each position of the sequence independently. It is a subnetwork composed of two linear layers with the application of a non-linear activation function ReLU in between.
\begin{align}
    \begin{split}
        \text{\ffnn{}(x)}&= Linear_{2}(ReLU(Linear_{1}(\text{z})) \\
    \label{eqn:ffnn_formulation}
    \end{split}
\end{align}
Note that, at the end of each block, the Transformer layer employs residual connections followed by a layer normalization to facilitate the flow of information and stabilize the training process.
\begin{align}
    \begin{split}
        \text{y}&= LayerNorm(\text{y\textsubscript{\mha{}}} + \text{\ffnn{}(y\textsubscript{\mha{}})})
    \label{eqn:ffnn_residual}
    \end{split}
\end{align}

% \subsection{Pruning}
% \todo{Move this part into a more suitable place}
% \textit{Pruning}~\cite{liang2021pruning} is a method to enhance the efficiency of deep learning models.
% Pruning can be
% further classified into \textit{weight pruning} and \textit{token pruning}~\cite{Kim2022LTP}. Weight pruning involves removing certain weights from the model, while token pruning eliminates unimportant tokens from the input. 

\subsection{Motivation}
\label{subsec:motivation}
\subsubsection{\underline{Impact on Energy Consumption}}
Language models for code require a significant amount of computational power to train and run. The number of parameters in these models has increased from millions to billions~\cite{Hou2023,Zhao2023}, necessitating the use of high-performance computing systems and specialized hardware, such as graphics processing units (\gpu{}s) and tensor processing units ({\textsc{tpu}}s).
The training of such models involves processing large amounts of data, often in the order of terabytes. This requires not only powerful processors but also high-bandwidth memory and storage systems. As a result, the energy consumed during the training phase is substantial, as it involves running these models for extended periods that can reach the magnitude of days. Table~\ref{tab:model_training_time} illustrates the training time of different models of various sizes.

\begin{table}[h!]
  \centering
  \caption{Training time of different language models for code.}
  \begin{tabular}{@{}lll@{}}
    \toprule
    Model & \# Parameters & Training Time \\
    \midrule
    UniXCoder~\cite{Guo2022} & 125M & 8 days \\
    CodeT5\textsubscript{\sc{base}}~\cite{Wang2021codet5} & 220M & 12 days \\
    InCoder~\cite{Fried2023} & 6.7B & 24 days \\
    % Add more rows as needed
    \bottomrule
    \label{tab:model_training_time}
  \end{tabular}
\end{table}

The energy consumption of language models is directly proportional to their computational requirements. As the models become more complex and the amount of data they process increases, the energy needed to power the computing systems also rises. This results in a higher carbon footprint of these models. Wang \etal{}~\cite{Wang2013Energy} have shown that the carbon footprint of pre-training the {\sc bert} model~\cite{Devlin2018} on a machine with one NVIDIA Quadro RTX$8000$ \gpu{} is $199$ kg CO\textsubscript{2}. This is equivalent to the same amount emitted by driving a vehicle for a distance of $820$ kilometers\footnote{Wang \etal{}~\cite{Wang2013Energy} report the distance in miles}~\cite{Wang2013Energy}. 
% \todo{refer table-1 in this paragraph. MS: I referred to it above given that it contains training time.}
Moreover, we are witnessing a notable trend of models increasing in size leading to even higher environmental impact.
For instance, recently, Lozhkov \etal{}~\cite{Lozhkov2024Starcoder} released StarCoder 2, a family of language models for code with sizes ranging from $3$B to $15$B parameters. They reported that the smallest model has resulted in $15K$ kg of CO\textsubscript{2} emission during pre-training.

\noindent
\subsubsection{\underline{Computational Cost of the Transformer Layer}} This section explores the computational complexity and memory requirements of both the \mha{} and \ffnn{} layers.

The computational cost of a neural layer can be measured in terms of the number of floating-point operations (\flops{}) performed by such layer. 
We calculate the theoretical operations count, \ie{} the upper bound when a sequence has the length of the maximum number of tokens. 
The \flops{} counting rules introduced by Hunger~\cite{Hunger2005floating} for matrix multiplications are used in our work. 
Specifically, the rule stating that the number of \flops{} of a matrix $A_{M \times N}$ multiplied with another matrix $B_{N \times L}$ is $2MNL-ML$. Note that here $(M\times N)$ and $(N\times L)$ represent the dimensions of the matrices $A$ and $B$, respectively.

\noindent
\textbf{FLOPs of the \mha{}:}  Let $n$ be the maximum number of tokens, $d_{\text{\mha{}}}$ the hidden dimension, and $h$ the number of heads. 
The first step includes three linear projections of the input $X_{n\times d_{\text{\mha{}}}}$ with the $W^{Q}_{d_{\text{\mha{}}}\times \frac{d_{\text{\mha{}}}}{h}}$, $W^{K}_{d_{\text{\mha{}}}\times \frac{d_{\text{\mha{}}}}{h}}$, and $W^{V}_{d_{\text{\mha{}}}\times \frac{d_{\text{\mha{}}}}{h}}$ weight matrices at each attention head. These operations involve matrix multiplication and addition (for the bias term). The total number of operations at this stage is calculated as follows.

\begin{equation}
\text{F\textsubscript{MatMul}} = 2 \times n \times d_{\text{\mha{}}} \times \left(\frac{d_{\text{\mha{}}}}{h}\right) - n \times \left(\frac{d_{\text{\mha{}}}}{h}\right)
\label{eqn:first_lin_proj}
\end{equation}

\begin{align}
\begin{split}
    \text{F\textsubscript{LinearProj}} &= h \times 3 \times \text{\text{F\textsubscript{MatMul}}} \\
&= 6 \times n \times d_{\text{\mha{}}}^{2} - 3 \times n \times d_{\text{\mha{}}}
\end{split}
\end{align}

The scaled dot-product attention is calculated by multiplying the resulting query matrix $Q_{n \times \frac{d_{\text{\mha{}}}}{h}}$ with the transpose key matrix $K^{T}_{\frac{d_{\text{\mha{}}}}{h} \times n}$ coupled with the application of the scaled \textit{softmax} function.

\begin{equation}
    \text{F\textsubscript{QK}} = 2 \times n \times \left(\frac{d_{\text{\mha{}}}}{h}\right) \times n - n^{2}
\end{equation}
\begin{equation}
    \text{F\textsubscript{SoftmaxScaling}} = 2 \times n^{2}
\end{equation}
\begin{align}
    \begin{split}
        \text{F\textsubscript{ScaledDotAttn}} {}&=  \text{F\textsubscript{QK}} +  \text{F\textsubscript{SoftmaxScaling}} \\
            &= 2 \times n^{2} \times d + h \times n^{2}
    \end{split}
\label{eqn:attn_prob_flops}
\end{align}
    
The next operation in the \mha{} layer is the multiplication of the attention probability matrices of each head $A_{n \times n}$ obtained in Equation~\ref{eqn:attn_prob_flops} with the value matrix $V_{n \times \frac{d_{\text{\mha{}}}}{h}}$ that was calculated from one of the linear projections performed earlier.
\vspace{-2mm}
\begin{align}
    \begin{split}
        \text{F\textsubscript{Attn}} &= h \times \left(2 \times n \times n \times \left(\frac{d_{\text{\mha{}}}}{h}\right) - n \times \left(\frac{d_{\text{\mha{}}}}{h}\right)\right) \\
        &= 2 \times n^{2} \times d_{\text{\mha{}}} - n \times d_{\text{\mha{}}}
    \end{split}
\end{align}

Finally, the outputs of all attention heads are concatenated and projected through a linear layer. The concatenation operation requires 0 operations, hence, the only \flops{} performed are from the matrix multiplication of the concatenated matrix $Z_{n \times d_{\text{\mha{}}}}$ with $W^{O}_{d_{\text{\mha{}}} \times d_{\text{\mha{}}}}$.

\begin{equation}
    \text{F\textsubscript{FinalProj}} = 2 \times n \times d_{\text{\mha{}}}^{2} - n \times d_{\text{\mha{}}}
    \label{eqn:final_mha}
\end{equation}

Summing the Equations~\ref{eqn:first_lin_proj}-\ref{eqn:final_mha}, we get the following formula to determine the \flops{} count of the attention layer.
\vspace{-1mm}
\begin{equation}
    \text{F\textsubscript{\mha{}}} = 8 \times n \times d_{\text{\mha{}}}^{2} + 4 \times n^{2} \times d_{\text{\mha{}}} - 4 \times n \times d_{\text{\mha{}}} + h \times n^{2}
    \label{eqn:total_flops_mha}
\end{equation}

\noindent
\textbf{FLOPs of the \ffnn{}:} The \ffnn{} takes the output of the \mha{}'s layer as input, scales it to a higher dimension than the \mha{}'s, applies a {\sc gelu} transformation~\cite{Hendrycks2016gaussian}\footnote{While the Transformer layer was stated earlier to use ReLU activation, the models used in this work replace it with the {\sc gelu} function. Note that this does not affect the calculations performed above.}, and scales it back to the original hidden dimension.

Following the same notation, let $n$ be the number of tokens in a sequence and $d_{\text{\ffnn{}}}$ the dimension of the \ffnn{} layer.
The \flops{} count for the first linear layer includes a multiplication between the output matrix of the \mha{} $M_{n \times d_{\text{\mha{}}}}$ with the weight matrix of this layer $W^{L}_{d_{\text{\mha{}}} \times d_{\text{\ffnn{}}}}$ proceeded by the application of the {\sc gelu} function.
\vspace{-1mm}
\begin{align}
    \begin{split}
    \text{F\textsubscript{FirstLinProj}} &= 2 \times n \times d_{\text{\mha{}}} \times d_{\text{\ffnn{}}} - n \times d_{\text{\ffnn{}}} \\
    \text{F\textsubscript{GELU}} &= n \times d_{\text{\ffnn{}}} \\
    \text{F\textsubscript{FirstLayer}} &= \text{F\textsubscript{FirstLinProj}} + \text{F\textsubscript{GELU}} \\
    &=  2 \times n \times d_{\text{\mha{}}} \times d_{\text{\ffnn{}}}
    \end{split}
    \label{eqn:flops_ffnn_first_layer}
\end{align}

The final operation projects the output to the original hidden dimension.
\vspace{-1mm}
\begin{equation}
    \text{F\textsubscript{SecondLayer}} = 2 \times n \times d_{\text{\ffnn{}}} \times d_{\text{\ffnn{}}} - n \times d_{\text{\ffnn{}}}
    \label{eqn:flops_ffnn_second_layer}
\end{equation}

The \ffnn{}'s dimension is usually set to be higher than the \mha{}'s. For instance, in the models that we have chosen for our experiments,  $d_{\text{\ffnn{}}}=4d_{\text{\mha{}}}$. Using this fact, and the Equations~\ref{eqn:flops_ffnn_first_layer} and \ref{eqn:flops_ffnn_second_layer} the total number of \flops{} is,
\vspace{-1mm}
\begin{equation}
    \text{F\textsubscript{\ffnn{}}} = 16 \times n \times d_{\text{\mha{}}}^{2} - n \times d_{\text{\mha{}}}
    \label{eqn:total_flops_ffnn}
\end{equation}

Although Equations~\ref{eqn:total_flops_mha} and~\ref{eqn:total_flops_ffnn} show that the computational cost of \mha{} increases quadratically compared to the \ffnn{}, it uses fewer \flops{} when the maximum sequence length is $1024$. Such length is an upper bound for many language models such as \cb{}~\cite{Feng2020}, \gcb{}~\cite{Guo2021}, and \unx{}~\cite{Guo2022} that are used to solve SE tasks. We now prove this statement analytically,

\begin{lemma}\label{flops_flops_gt_flops_mha} For every conceivable input sequence fed into a Transformer model with hyperparameters same as \cb{}, \gcb{}, and \unx{}, it holds true that $$\text{F\textsubscript{\ffnn{}}} > \text{F\textsubscript{\mha{}}}.$$
\end{lemma}

\begin{proof}
Let us fix $d:=d_{\text{\mha{}}}>0$, $h>0$ and consider the following function:

$$f(n) = \text{F\textsubscript{\ffnn{}}}(n) - \text{F\textsubscript{\mha{}}}(n) = n(8d^2 +3d) - n^2(4d+h)$$

This function
\begin{itemize}
    \item is quadratic
    \item has a maximum $(f''(n) <0)$.
    \item has two roots $n_1=0$ and $n_2=\frac{8d^2 +3d}{4d + h}$
\end{itemize}
\noindent therefore $f(n)>0$ or $\text{F\textsubscript{\ffnn{}}} > \text{F\textsubscript{\mha{}}}$ when $n\in(0,n_2)$. In other words, if the sequence length $n\leq\lfloor \frac{8d^2 +3d}{4d + h} \rfloor$, then $\text{F\textsubscript{\ffnn{}}} > \text{F\textsubscript{\mha{}}}.$

For \unx{}, \gcb{}, and \cb{}, we have that $d=768$ and $h=12$. Thus,
$$\left\lfloor \frac{8d^2 +3d}{4d + h} \right\rfloor = 1530$$
Furthermore, the maximum sequence lengths are 1024, 512, and 512, for \unx{}, \gcb{}, and \cb{} respectively. Therefore, it will consistently hold true that $n < \lfloor \frac{8d^2 +3d}{4d + h} \rfloor$, leading to the conclusion that $\text{F\textsubscript{\ffnn{}}} > \text{F\textsubscript{\mha{}}}$.

\end{proof}

\begin{wrapfigure}{r}{0.5\textwidth}
  \input{figures/theo_empr_flops_count}
\end{wrapfigure}
This is further demonstrated empirically when measuring the number of \flops{} performed by both of these layers in \cb{}~\cite{Feng2020} using the test set of the \devign{}~\cite{Zhou2019devign} dataset as shown in Figure~\ref{fig:emp_codebert}.

\noindent
Therefore, reducing the sequence length before feeding it to the \ffnn{} layer is a potential way to reduce such complexity.

\section{Approach}

In this section, we present the pruning procedure adopted in this study in each Transformer layer and elaborate on the design choices.

\begin{figure}[H]
    \centering
    \includegraphics[width=\textwidth]{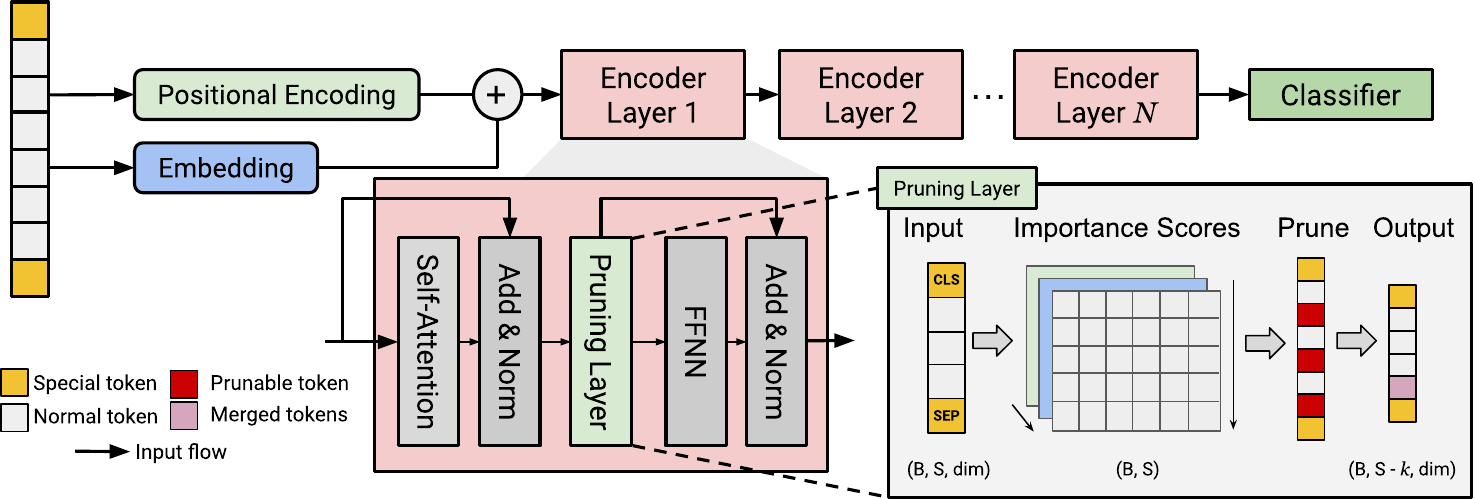}
    \caption{Overview of a Transformer encoder-based model using \alpine{}. Tokens that are highlighted in yellow represent special tokens such as \texttt{\textbf{[CLS]}} and \texttt{\textbf{[SEP]}}. Whenever applicable, we include the tensors dimensions for better clarity. In the figure, $B$ refers to the batch size, $S$ is the sequence length, and $dim$ is the hidden dimension of the model. The number of tokens that are pruned is $k$ which would differ from one layer to another.}
    \label{fig:overview}
\end{figure}
As shown in Figure~\ref{fig:overview}
and Algorithm~\ref{alg:transformer_implementation}, 
the pruning layer is inserted between the \mha{} and \ffnn{} layers. 
% \todo{Please ensure that we have defined MHA and FFNN earlier. This is already done in the Background section (MS).}
As highlighted in~\cref{flops_flops_gt_flops_mha}, the feed-forward layer represents a bottleneck in terms of computational costs and memory consumption~\cite{Ganesh2021} caused by
the projection of the attention layer's output to a higher dimension in the \ffnn{}.
% given its projection of the output of the attention layer into a higher dimension. 
Hence, performing pruning before this layer would reduce the effect of such a bottleneck.

% \todo{we may refer the algorithm using the line no in the relevant steps. MS: Done}
Algorithm~\ref{alg:pruning_operation} presents the steps to perform the pruning.
The first step is to assign a score to each token to determine its \textit{importance} within the sequence. In this work, we use \textit{attention} based scoring~\cite{Kim2022LTP}, which quantifies this importance using the attention probabilities output by the \mha{} layer. Specifically, in Lines 4--11 we take the mean across the attention matrices of the attention heads of the \mha{} layer and then perform a column-wise mean to get a distribution of importance scores as illustrated by Equation~\ref{eq:importance_score}, where $h$ is the number of attention heads, $\mathbf{A}$ is the attention matrix of an attention head and $n$ is the list of tokens \textit{excluding} the special tokens \clstok{}, \septok{}, and \padtok{}.
\begin{equation}
    s(t_{i}) = \frac{1}{h\cdot n}\sum_{j=1}^{h}\sum_{k=1}^{n}\mathbf{A}_{jki}
\label{eq:importance_score}
\end{equation}

% \begin{itemize}
%     \item Describe the overview figure in detail. Mention how special tokens are excluded from the importance score calculation (you can cite your work from ICSE'24 NIER, Clark's and Kobayashi's from NLP). The reason is because CLS and SEP get much higher attention scores and SEP get 0 scores. Including such values would influence the score distribution and the results.
%     \item Mention how the bound is calculated.
%     \item Create a pseudocode algorithm for the pruning function. However, before doing so introduce the notation: What is A, whats a mask and so on.
%     \item Mention that are three ways we apply this pruning: across all layers, even indexed layers and odd indexed layers.
% \end{itemize}

% Find a way to differentiate between variable names, constants and function calls.
\begin{algorithm}[ht!]
\fontsize{8}{10}\selectfont
\caption{Perform pruning by updating the tokens mask}
\begin{algorithmic}[1]
\Require \State $A$: Attention probabilities of each attention head.
\State $M$: Initial mask of a sequence $L$.
\State $\alpha$: Width of the bounds used for pruning tokens.
\Ensure $M^{U}$: Updated update mask for prunable tokens.
\State $A \gets \texttt{mean}(A)$ \Comment{Take the mean across all attention heads}
\State $scores \gets \texttt{mean}(A)$ \Comment{Take the mean across the columns}
\State $sep\_idx \gets \texttt{sum}(M) - 1$ \Comment{Get the index of the SEP token}
\State $scores[0] \gets \mathtt{NaN}$ \Comment{Set CLS to NaN}
\State $scores[sep\_idx] \gets \mathtt{NaN}$ \Comment{Set SEP to NaN}
\State $seq\_len \gets \texttt{length}(M)$
\State $pad\_idx \gets sep\_idx + 1$
\State $scores[pad\_idx: seq\_len - 1] \gets \mathtt{NaN}$ \Comment{Set PAD to NaN}
\State $\mu \gets \texttt{meannan}(scores)$ \Comment{Calculates the mean ignoring NaNs}
\State $\sigma \gets \texttt{stdnan}(scores)$ \Comment{Calculates the SD ignoring NaNs}
\State $max \gets \mu + \alpha \cdot \sigma$
\State $min \gets \mu - \alpha \cdot \sigma$
\State $keep\_tokens \gets (scores \geq min) \wedge (scores \leq max) \wedge (\neg \texttt{isnan}(scores))$
\State $M^{U} \gets \texttt{zeros}(\texttt{shape}(M))$ \Comment{Initalize the updated mask with zeros}
\State $M^{U}[keep\_tokens] \gets 1$
\State $M^{U}[0] \gets 1$ \Comment{Set CLS to 1}
\State $M^{U}[sep\_idx] \gets 1$ \Comment{Set SEP to 1}
\State \Return $M^{U}$
\end{algorithmic}
\label{alg:pruning_operation}
\end{algorithm}
The reason we exclude the \clstok{} and \septok{} tokens from pruning is two-fold. First, 
% the considered SE tasks are classification tasks. To perform the classification, 
the \clstok{} token representation is extracted from the output of the final transformer encoder layer and fed into a linear classification layer as done previously~\cite{Shi2023, Lu2021CodeXGLUE}. 
% \todo{check comments. MS: Done.}
Hence, such a token should be kept throughout the encoder pipeline. Second, it has been shown that these two tokens usually receive the highest attention scores \cite{Clark2019, Kobayashi2020, Saad2023} in natural and programming languages, which would skew the importance scores distribution. As for the padding token, its attention score is 0, which again would affect the score distribution, especially for shorter sequences.

\begin{algorithm}[ht]
\fontsize{8}{10}\selectfont
\caption{Modified Transformer Encoder Layer}
\begin{algorithmic}[1]
\Require 
\State Input: Input from the previous layer.
\State $M$: Input mask.
\State $\alpha$: Width of the bounds used for pruning tokens.
\State $merge$: A boolean indicating whether pruned tokens should be merged.
\Ensure output\textsuperscript{FFNN}: Output of the FFNN layer to the subsequent layer.
\Ensure $M^{U}$: Input mask for the subsequent layer.
\State output\textsuperscript{MHA}, $A$ $\gets$ MHA(Input, $M$)
\State output\textsuperscript{MHA} $\gets$ LayerNorm(Input + output\textsuperscript{MHA})
\State \textcolor[rgb]{0,0.5,0}{$M^{U}$ $\gets$ Prune($A$, $M$, $\alpha$)}
\State \textcolor[rgb]{0,0.5,0}{output\textsuperscript{MHA}, $M^{U}$ $\gets$ RepackTensor(output\textsuperscript{MHA}, $M^{U}$, $merge$)}
\State output\textsuperscript{FFNN} $\gets$ FFNN(output\textsuperscript{MHA})
\State output\textsuperscript{FFNN} $\gets$ LayerNorm(output\textsuperscript{FFNN} + output\textsuperscript{MHA})
\State \Return output\textsuperscript{FFNN}, $M^{U}$
\end{algorithmic}
\label{alg:transformer_implementation}
\end{algorithm}

Once we obtain the importance scores distribution, we keep the tokens that are within a specific range. We set the range to be $R=[\mu - \alpha \cdot \sigma \; , \; \mu + \alpha \cdot \sigma ]$, where $\mu$ and $\sigma$ are the mean and standard deviation of the importance scores distribution, and $\alpha$ is a hyperparameter that defines the \textit{width} of the lower and upper bounds as shown in Lines 12-16. The smaller $\alpha$ gets, the tighter the window, the less tokens are kept. This hyperparameter controls the aggressiveness of pruning.

Based on a preliminary investigation, we observed that the distribution of importance tokens follows roughly a \emph{leptokurtic distribution} (\ie{} distribution with low standard deviation and a high peak) at each layer. By removing tokens outside the defined range in such distribution, we can preserve the most informative and representative data while reducing computational complexity. 

Given the importance scores and the range $R$, we create a new mask $M^{U}$ that indicates the tokens that should be pruned and those that should be kept (Lines 17-20),

\[
M^{U}[i] = 
\begin{cases}
    0 & \text{if } t_{i} = \text{\padtok{}} \text{ or } s(t_{i}) \notin R, \\
    1 & \text{otherwise.}
\end{cases}
\]
Algorithm~\ref{alg:transformer_implementation} shows the modified implementation of the Transformer layer. The highlighted Lines 7 and 8 represent the modifications that we have introduced. This also demonstrates the ease of integrating \alpine{} with minimal effort.

Using the updated mask from the previous step, we either remove the rows from the \mha{} output matrix that correspond to the tokens to be pruned
% or merge them into a single row 
as shown in Line 8, or merge them into one row.
By default, we do the merging to minimize the information loss incurred by pruning, especially that of tokens that are above the upper bound of the importance interval. Next, the matrix's dimensions are reduced. We also calculate the final mask since this operation is performed on a batch of sequences to guarantee that all sequences are of equal length.

%\todo{We should introduce acronyms / names for thes variants that can be referred to in Sec IV (assuming they are used there).}
We experiment with three variants of pruning. The first one involves performing pruning across \emph{all layers} of the encoder. In the second setting, it is performed only for \emph{even-indexed layers} (\ie{} $0, 2, \dots{} 10$). Whereas the final one involves pruning at \emph{odd-indexed layers} (\ie{} $1, 3, \dots{} 11$).
\section{Experiments}

The goal of this study is to investigate the effect of \alpine{} on language models of code when applied to SE tasks. The research aims to assess the balance between efficiency, from computational and environmental points of view, and model accuracy.
Toward this goal, we formulate the following research questions. The rest of the section elaborates on the experiments designed to answer the research questions.

\smallskip
\begin{itemize}
    \item [\textbf{RQ1.}] 
    \textit{To what extent does pruning on language models for code save computational resources?}\\
    % \textit{How much computation is saved when introducing pruning in language models for code?}\\
    Intuitively, the less data a model processes, the more computation is saved. With this question, we wish to confirm and quantify these computational savings compared to the baseline, non-pruned model.
    \item [\textbf{RQ2.}]
    \textit{What is the impact of the pruning technique on the performance of language models for code on various SE tasks?}\\
    Naturally, pruning tokens would result in some information loss. In this research question, we aim to measure the extent of performance drop, in terms of accuracy, for example, that might occur, if any.
    \item [\textbf{RQ3.}] 
    \textit{What is the effect of merging prunable tokens as opposed to completely dropping them?}\\
    The proposed approach allows one to choose whether pruned tokens are to be partially kept by merging their representation into a single row or entirely removed from the sequence. In this research question, we investigate the impact of such a design choice.
    \item [\textbf{RQ4.}] 
    \textit{What is the impact of the computing environment on the efficiency of the pruning method?}\\
    Finally, an important reason behind model simplification is to allow computationally demanding models to run on relatively lesser-capable \gpu{}s (\eg{} consumer-grade \gpu{}s) compared to those used in high-performance computing clusters (\eg{} NVIDIA V$100$ or A$100$). This exploration would ensure wider accessibility and make using these models more practical.
\end{itemize}

\subsection{Software Engineering Tasks}
In this study, we choose two substantially different tasks from the {\sc CodeXGLUE} benchmark~\cite{Lu2021CodeXGLUE} for our experiments.
Both tasks aim to improve code quality, developer productivity, software reliability, and maintainability.

\smallskip
\noindent
\textbf{Code clone detection:} Clone detection identifies similar or identical code fragments, helping maintain code quality, reduce bugs, and enable refactoring. In this work, we used the filtered version of the BigCloneBenchmark dataset~\cite{Svajlenko2014} released by Wang \etal{}~\cite{Wang2020CodeClone}. 
This binary classification problem aims to predict whether two Java code snippets are clones. The dataset is composed of $901K$ training pairs and $415K$ pairs for testing. 
% Given such a considerable size, 
We randomly take a subsample of $90K$ pairs for training and $4K$ for evaluation;
such random sampling is quite common in this domain~\cite{Shi2023, Yang2022}. We use the same subsamples for all experiments during training and evaluation.

\noindent
\textbf{Defect prediction:} It identifies bugs, vulnerabilities, and defects early, reducing the cost and effort of manual code reviews and testing. We used Zhou \etal{}~\cite{Zhou2019devign}'s dataset that was curated from a set of four large-scale open-source C repositories. The authors manually annotated the code snippets and covered multiple types of defects such as resource leaks and use-after-free vulnerabilities. 
We use the default pre-split sets of $21,854$ samples for training and $2,732$ for validation and testing.

% Both tasks aim to improve code quality, developer productivity, software reliability, and maintainability.

\subsection{Model selection}
We select the following models given that they represent state-of-the-art encoder-only models in solving the previously mentioned tasks as of conducting this work~\cite{CodeXGLUELB, Shi2024, Niu2023, Zeng2022}.

\smallskip
\noindent
\textbf{CodeBERT \cite{Feng2020}:} It is a pre-trained language model that has the same architecture and a similar pretraining strategy as {\sc RoBERTa}~\cite{Lui2019RoBERTa}.
It comprises twelve Transformer layers,
each containing a 12-headed \mha{} layer and has a hidden dimension of $768$. It was trained on bimodal data consisting of pairs of natural language and programming language across six programming languages from the CodeSearchNet dataset~\cite{Husain2019}.

\noindent
\textbf{GraphCodeBERT \cite{Guo2021}:} \gcb{} extends \cb{} by including the data flow information within the source code input. It also includes an additional pre-training task where the objective is to predict whether a data flow edge exists between two nodes.

\noindent
\textbf{UniXCoder~\cite{Guo2022}:} 
\unx{} leverages a multi-modal input consisting of natural language, source code, and the flattened sequence of the abstract syntax tree of the code snippet. It is similar to the two other models that share the same internal architecture. The difference here (aside from the pre-training objectives and input representation during pre-training) is that \unx{} allows for a larger context length, $1024$ compared to $512$ for (Graph)CodeBERT.

We use the same model architecture in both tasks. Specifically, we add a classifier in the form of a dense layer on top of each encoder. It takes as input the last hidden representation of the \clstok{} token.

\subsection{Evaluation Metrics}
To assess the predictive performance of the models on the aforementioned tasks, we use the same evaluation metrics that were reported by Lu \etal{}~\cite{Lu2021CodeXGLUE} in the CodeXGLUE benchmark. Specifically, we calculate the \textit{accuracy} for the Devign dataset and \textit{F1}-score for the BigCloneBenchmark.
As for the computational efficiency, we report the number of floating points operations (\flops{}) and \textit{Throughput}. As stated in Section~\ref{subsec:motivation} the number \flops{} quantifies the computational complexity. 
% \todo{complexity? or.. capacity. MS: Complexity. Capacity is characterized by the number of floating points per second, which is a metric to rank GPUs.}. 
A model with a higher {\sc flop} count requires more computing resources. 
The goal is to reduce such a count while maintaining the predictive performance as much as possible. 
The throughput refers to the number of input samples a model can process in a second. This metric is especially relevant during inference and model deployment.

\subsection{Experimental Setting}
We ran the experiments on two machines. 
% The first machine is used for model training and evaluation. 
The first machine has an {\sc amd} Milan 7413 {\sc cpu}, $64${\sc gb} of {\sc ram} and an NVIDIA A$100$ with $40${\sc gb} of v{\sc ram}. The second machine has an Intel i$7-8700$ {\sc cpu}, $32${\sc gb} of {\sc ram}, and an NVIDIA RTX$2080$ \gpu{} with $8${\sc gb} of v{\sc ram}. We use the same hyperparameters set in the {\sc CodeXGLUE} benchmark for each task.

\begin{table}[h]
\centering
\caption{Comparison of the pruning methods for the two SE tasks. The highest \textit{accuracy} or F1-score 
% \todo{Check the correct hyphenation of F1-score. MS: Done.} 
is \textbf{\underline{underlined and bolded}}. The pruning method with the highest accuracy is highlighted with a $\bigstar$. In the \textit{FLOPs} column, the values between parenthesis indicate the speedup ratio achieved by the pruning setting compared to the non-pruned model. The unit of throughput (TP) is \textit{samples per second} and the reported values were measured on an NVIDIA RTX$2080$. All metrics are calculated using the evaluation sets.}
\begin{subtable}[H]{\textwidth}

\resizebox{\textwidth}{!}{
\begin{tabular}{lccclcclcc}
\toprule
\multicolumn{4}{c}{\cb{}} & \multicolumn{3}{c}{\gcb{}} & \multicolumn{3}{c}{\unx{}} \\
\cmidrule(lr){1-4} \cmidrule(lr){5-7} \cmidrule(lr){8-10}
\makecell[l]{Pruning Method} & Accuracy & FLOPs ($\times 10^9$) & TP & Accuracy & FLOPs ($\times 10^9$) & TP & Accuracy & FLOPs ($\times 10^9$) & TP \\
\midrule
Baseline (No Pruning) & $64.02$\% & $48.29$ & $69.34$& $\mathbf{\underline{64.7}}$\% & $48.29$ & $69.19$ & $\mathbf{\underline{66.54}}$\% & $48.29$ & $69.37$ \\  \hdashline \addlinespace
All Layers & $62.3$\% & $20.97 \, (\times 2.3)$ & $98.03$ & $61.49$\% & $23.84\,(\times2.02)$ & $91.29$ & $65.44$\% & $27.45\,(\times1.75)$ & $81.61$ \\
\makecell[l]{Even Indexed \\ \hspace{1mm}$L = 0,2,\dots{}10$} & $64.09$\% & $33.25\,(\times1.44)$ & $81.75$ & $64.02$\%\textsuperscript{$\bigstar$} & $34.51\,(\times1.39)$ & $80.07$ & $65.62$\%\textsuperscript{$\bigstar$} & $37.003\,(\times1.3)$ & $74.73$ \\
\makecell[l]{Odd Indexed \\ \hspace{1mm}$L = 1,3,\dots{}11$} & $\mathbf{\underline{64.68}}$\%\textsuperscript{$\bigstar$} & $36.59\,(\times1.31)$ & $74.22$  & $63.79$\% & $36.66\,(\times1.31)$ & $74.93$ & $65.59$\% & $39.01\,(\times1.23)$ & $72.03$ \\
\bottomrule
\end{tabular}
}
\caption{Defect Prediction (\devign{})}
\end{subtable}

\begin{subtable}[h]{\textwidth}
\resizebox{\textwidth}{!}{
\begin{tabular}{lccclcclcc}
\toprule
\multicolumn{4}{c}{\cb} & \multicolumn{3}{c}{\gcb} & \multicolumn{3}{c}{\unx} \\
\cmidrule(lr){1-4} \cmidrule(lr){5-7} \cmidrule(lr){8-10}
\makecell[l]{Pruning Method} & F1 & FLOPS ($\times 10^9$) & TP & F1 & FLOPS ($\times 10^9$) & TP & F1 & FLOPS ($\times 10^9$) & TP \\
\midrule
Baseline (No Pruning) & $\mathbf{\underline{93.09}}$\% & $96.6$ & $34.12$ & $91.98$\% & $96.5$ & $33.85$  & $95.07$\% & $96.5$ & $34.22$ \\  \hdashline \addlinespace
All Layers & $90.57$\% & $48.2\, (\times 2)$ & $51.71$  & $90.74$\% & $47.5\, (\times 2.03)$ & $51.54$ & $94.7$\% & $57.7\, (\times 1.67)$ & $47.21$ \\
\makecell[l]{Even Indexed \\ \hspace{1mm}$L = 0,2,\dots{}10$}  & $90.39$\% & $65.7\, (\times 1.47)$ & $44.88$ & $\mathbf{\underline{92.13}}$\%\textsuperscript{$\bigstar$} & $68.5\, (\times 1.40)$ & $43.10$  & $94.5$\% & $80.4\, (\times 1.2)$ & $36.32$ \\
\makecell[l]{Odd Indexed \\ \hspace{1mm}$L = 1,3,\dots{}11$} & $92.17$\%\textsuperscript{$\bigstar$} & $68.5\, (\times 1.41)$ & $42.23$  & $91.46$\% & $73.9\, (\times 1.30)$ & $39.05$  & $\mathbf{\underline{95.47}}$\%\textsuperscript{$\bigstar$} & $81.5\, (\times 1.18)$ & $35.87$ \\
\bottomrule
\end{tabular}
}
\caption{Code Clone Detection (\bcb{})}
\end{subtable}
\label{tab:results_pruning}
\end{table}
\begin{figure*}[ht]
  \centering
  \begin{subfigure}{\columnwidth}
    \centering
    \includegraphics[width=0.3\textwidth]{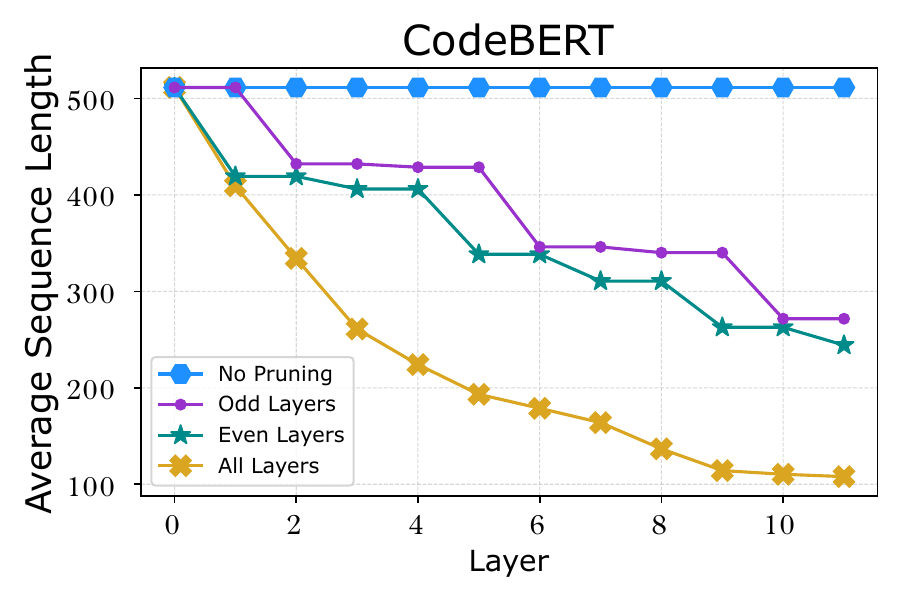}
    \includegraphics[width=0.3\textwidth]{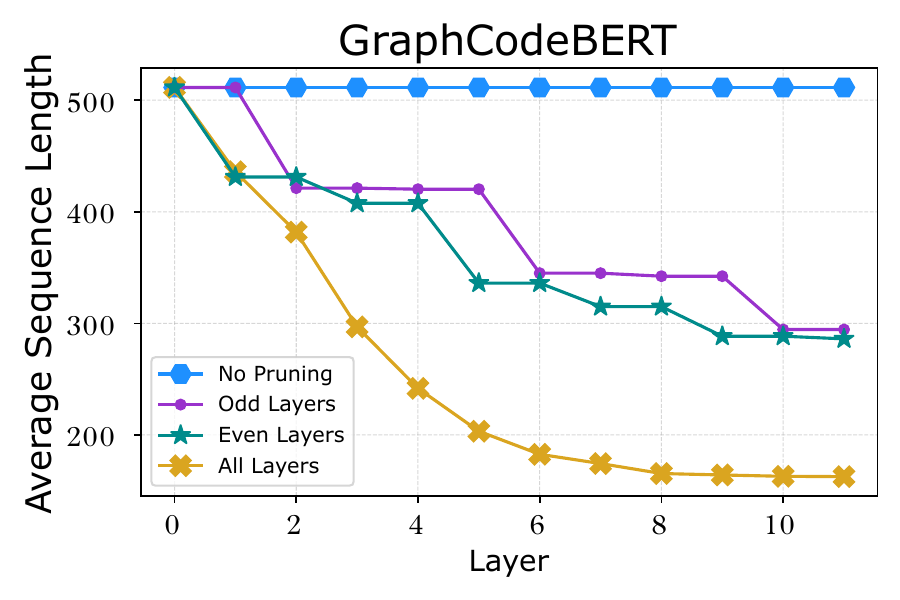}
    \includegraphics[width=0.3\textwidth]{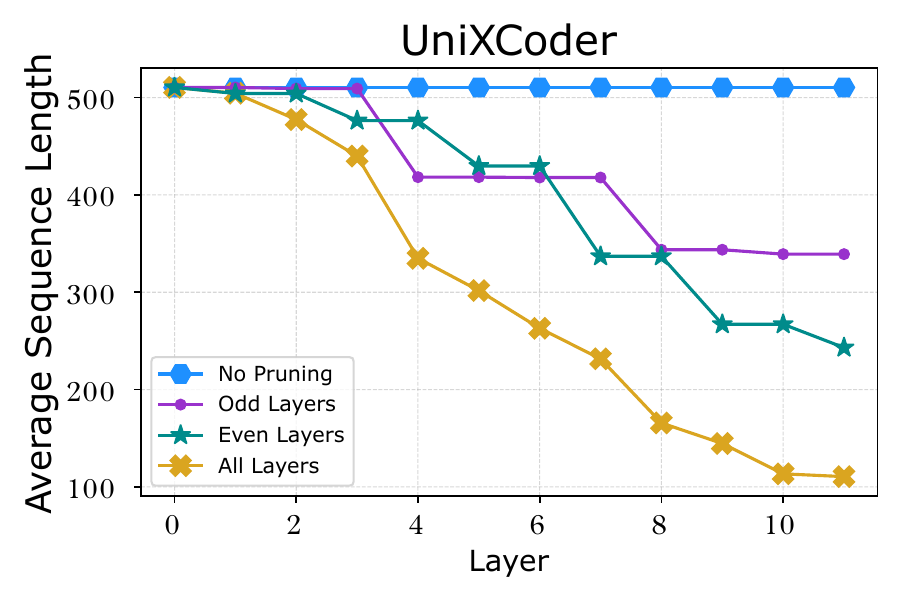}
    \vspace{-2mm}
    \caption{Defect Prediction (\devign{})}
  \end{subfigure}
  \vspace{-2mm}
  \begin{subfigure}{\columnwidth}
    \centering
    \includegraphics[width=0.3\textwidth]{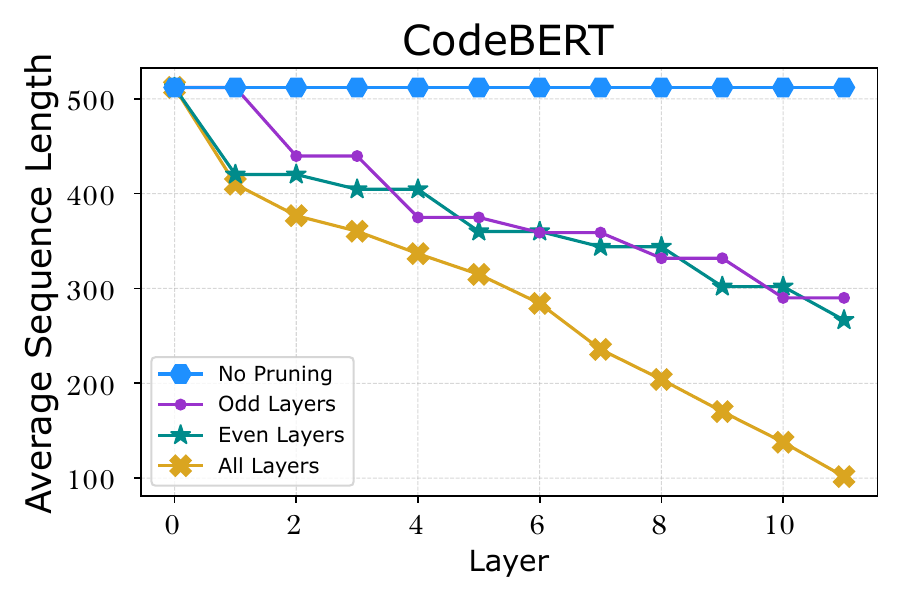}
    \includegraphics[width=0.3\textwidth]{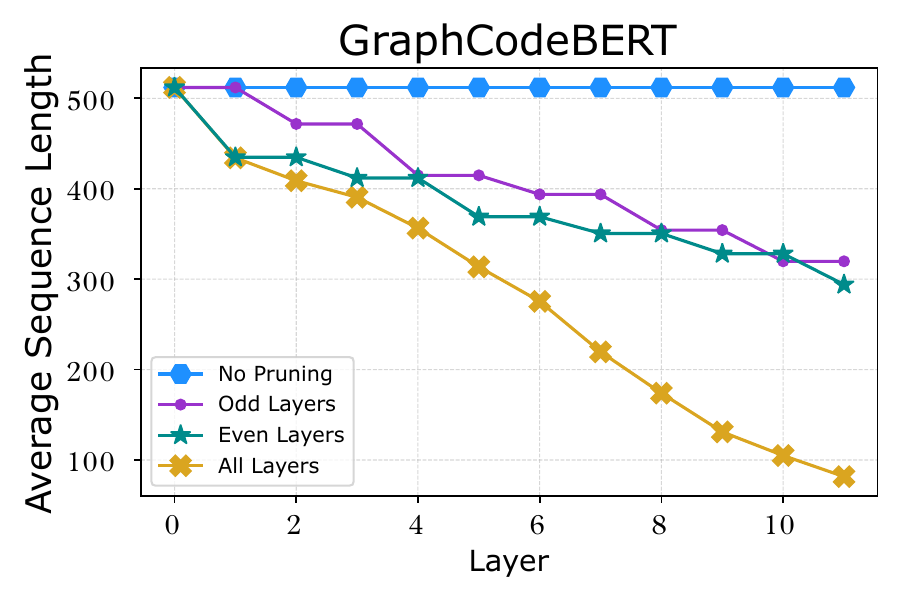}
    \includegraphics[width=0.3\textwidth]{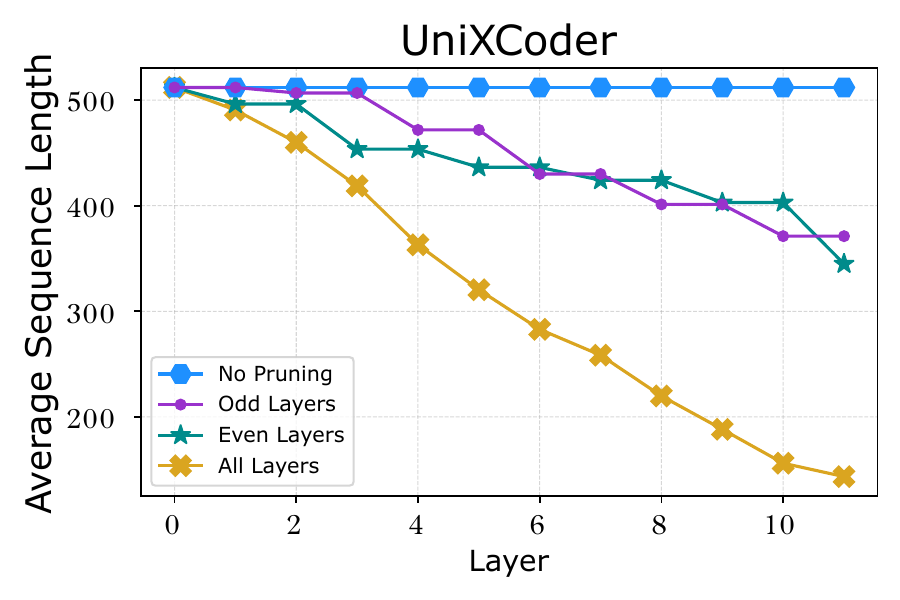}
    \vspace{-2mm}
    \caption{Code Clone Detection (\bcb{})}
  \end{subfigure}
  \caption{The progressive average reduction in sequences' lengths as the input traverses through the layers of each model. The plots are the result of 
  % The results are reported by performing 
  a forward pass across the whole evaluation set of each dataset with a batch size of 8. 
  % The sequences are padded to the longest sequence in the batch.
  }
  \label{fig:seq_len_progression}
\end{figure*}
\section{Results}

% This section presents the obtained results and discussion.

\subsection{\textbf{RQ1: Computational Efficiency}}
\label{results:rq1}
\noindent
\textit{\underline{Impact on FLOPs Count}}: Table~\ref{tab:results_pruning} highlights the computational speedup offered by \alpine{} across the three models on the SE tasks. 
Performing pruning across all layers allows for the most computational savings. 
On average\footnote{Given the values are ratios, we report the geometric mean.}, it reduces the \flops{}\footnote{Calculated using \href{https://github.com/facebookresearch/fvcore}{\texttt{fvcore}}.} count by $\times 2.01$ on the \devign{} dataset and $\times 1.9$ and \bcb{} compared the baseline models (\ie{} without any pruning). 
Specifically, \cb{} and \gcb{} exhibit roughly the same gain, where these models consume roughly half \flops{}.
% , where these models exhibit twice less computation. 
\unx{} also show significant gain, though slightly smaller ($\times 1.67$ and $\times 1.75$), in terms of \flops{}. 
% However, this gain seems to be reduced when leveraging \unx{}. It achieves $\times 1.75$-$\times 1.67$ fewer \flops{} than the original count across the two tasks.

Interestingly, performing pruning at even indexed layers seems to yield lower operations compared to pruning at odd indexed layers. A possible explanation for this behavior is that there is a considerable drop in the number of tokens at the 0\textsuperscript{th} layer which would have an earlier cascading effect if we were to perform this at a later layer.

In Figure~\ref{fig:seq_len_progression}, we plot the average sequence length across all the twelve layers during inference. When pruning across all layers on the \devign{} dataset, the sequence length plot exhibits an exponential-like decay on \cb{} and \gcb{}. On the other hand, this reduction follows roughly a linear trend on \bcb{} across all models. What is interesting is that at the final layer, the sequences are quite compressed. Specifically, the highest average sequence length at the last layer ($162$) is obtained when using \gcb{} on the defect prediction task which is 
$\sim\times 3$ less than the \textit{none} pruned average length of $512$.
In addition, aligned with the results of Table~\ref{tab:results_pruning}, we can see that pruning at even-indexed layers yields shorter sequences than odd-indexed ones. Across all task-model combinations and all layers (except for $2$-$4$ layer), \alpine{} produces more compressed sequences when applied at these layers compared to their odd-indexed counterpart. 

\smallskip
\noindent
\textit{\underline{Impact on Throughput}}: Pruning across all layers consistently yields the highest throughput improvement for all models and tasks. For instance, \cb{}'s throughput increases from $69.34$ to $98.03$ ($41.4$\% improvement) for the defect prediction task and from $34.12$ to $51.71$ ($51.6$\% improvement) for code clone detection.
Moreover, applying pruning in even-indexed layers generally results in better throughput compared to pruning odd-indexed layers. This trend is evident across all models and tasks, with even-indexed layer pruning providing an average throughput improvement of $16.7$\% compared to no pruning, while odd-indexed layer pruning offers an average improvement of $9.3$\%.
% \todo{we may have a rationale why the higher throughput makes sense when pruned. MS: Done (added at the end, right before memory footprint analysis.)}

Another observation is that the extent of improvement varies among the models. \gcb{} exhibits the highest average throughput improvement of $30.8$\% across all pruning methods and tasks, followed by \cb{} with $28.5$\% and \unx{} with $20.1$\%.
In general, the models achieved higher throughput on the \devign{} dataset, regardless of the pruning method applied. The average throughput for this task is $79.4$ samples per second, while for the code clone detection task, it is $41.5$ samples per second.

The increase in throughput is attributed to \alpine{}'s ability to compress sequences. A smaller sequence length allows for more efficient parallel processing and reduces the computational overhead associated with longer sequences, resulting in higher throughput and faster processing times.

\smallskip
\noindent
\textit{\underline{Impact on Memory Footprint}}: To further evaluate \alpine{}'s impact from a computational efficiency standpoint, we report the \gpu{} memory consumption. Due to space constraints, we only report the figures regarding pruning across all layers.

\begin{wrapfigure}{r}{0.55\textwidth}

    \raggedright
    \includegraphics[width=0.55\textwidth]{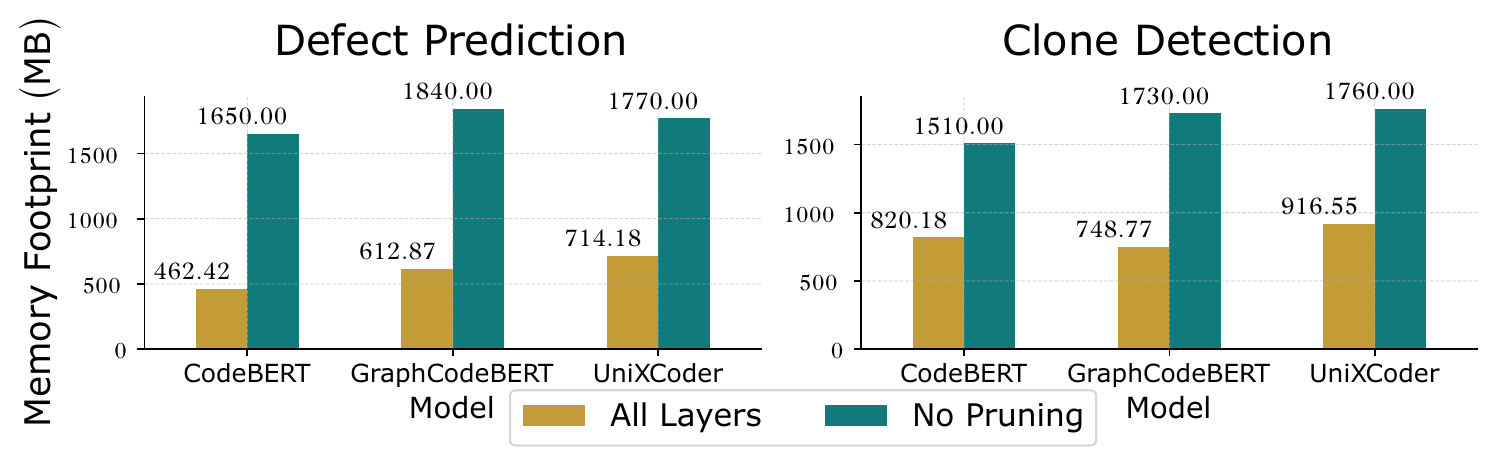}
    \caption{Comparison of the \gpu{} memory footprint between pruned and non-pruned models across the tasks. The measurements were conducted on the evaluation sets of each dataset during inference.}
    \label{fig:memory_footprint}

\end{wrapfigure}

From the plot in Figure~\ref{fig:memory_footprint}, we see that \alpine{} significantly reduces memory consumption across all models and tasks. \cb{} exhibits the highest percentage of memory reduction for the defect detection task ($72$\%), while \gcb{} and \unx{} show more consistent reductions across both tasks ($56.7$\% to $66.7$\%).

This reduction leads to improved cost efficiency, faster data transfer, and power efficiency. By requiring less memory, models can be deployed on more affordable \gpu{}s with smaller memory capacities, reducing hardware costs. Additionally, a smaller memory footprint results in faster data transfer between the \cpu{} and \gpu{}; the additional speed, in turn, increases throughput as demonstrated earlier. Moreover, 
% as \gpu{}s consume significant power, which increases with memory usage,
minimizing the model's memory footprint contributes to reduced power consumption,  which is crucial in constrained environments.

\vspace{-2mm}
\smrybx{\textbf{Summary}: \alpine{} achieves a high compression rate of sequences under different modes. As a result, it allows language models of code to be more computationally efficient. This is demonstrated by the $\times 2$ reduction in \flops{} count, up to $51.6$\% increase in model throughput, and up to $58.1$\% decrease in memory consumption.}

\subsection{\textbf{RQ2: Impact on Performance}}

In Table~\ref{tab:results_pruning}, we report the predictive performance of the pruned models. 
On the \devign{} dataset, performing pruning at odd-indexed layers leads to a slight improvement when using \cb{}, with a $0.66$\% increase in accuracy compared to that of the non-pruned version. On the other hand, adding the pruning layer at even layers is more efficient for \gcb{} and \unx{}, where they were able to maintain $98.94$\% and $98.61$\% of the original accuracy scores, respectively.
In the clone detection task, pruning at even-indexed layers achieves the highest F1-scores 
compared to other pruning strategies. It even managed to improve the clone detection performance by increasing the F1-scores of \gcb{} and \unx{} by $0.15$\% and $0.4$\%.

Applying \alpine{} at all layers
not only leads to $\times 2$ less computation 
but also results in a high level of predictive performance for all three models. \unx{} exhibits the highest performance both before and after pruning, retaining $98.3$\% of its original accuracy. \cb{} and \gcb{} also demonstrate strong performance retention, preserving $97.3$\% and $94.4$\% of their original accuracy, respectively. A similar observation can also be made for the results of code clone detection. All models maintain a high level of performance after pruning, with \unx{} achieving the highest F1-score retention at $99.6$\% of its original score. \gcb{} and \cb{} also demonstrate strong performance, retaining $98.7$\% and $97.3$\% of their original F1-scores, respectively.

Among the three models, \unx{} shows the highest resilience to pruning, keeping an average of $99$\% of its original performance across both tasks and all pruning methods. 
%\vspace{-0.5mm}
\smrybx{\textbf{Summary}: \alpine{} maintains a high percentage of the original performance across all models and tasks, with some instances even surpassing the performance of the non-pruned models. On average, the pruned models retain $98.15$\% of the performance of their non-pruned counterparts. The effectiveness of the actual pruning strategy, depends on the model and most importantly on the task. On \devign{}, integrating \alpine{} in even-index layers has led to the highest performance retention, whereas on \bcb{} such an outcome was achieved by odd-indexed layers.}

\subsection{\textbf{RQ3: Effect of Merging Prunable Tokens on Model Accuracy}}
\begin{table}[ht!]%[H]
\centering
\caption{The effect of token merging on models' accuracy and F1-scores. \textit{w/ m} refers to the setting where token merging is enabled whereas \textit{w/o m} indicates pruning without merging. As a reminder the performance metric for defect prediction is accuracy and F1 for code clone detection.}
\resizebox{\columnwidth}{!}{
\begin{tabular}{cc|c|c|c|c|c|c|c|c|c|c|c}
\toprule
\multirow{5}{*}{Model} & \multicolumn{12}{c}{Task} \\
\cmidrule(lr){2-13}
& \multicolumn{6}{c|}{Defect Prediction} & \multicolumn{6}{|c}{Code Clone Detection} \\
\cmidrule(lr){2-13}
& \multicolumn{2}{c|}{All} & \multicolumn{2}{c|}{Even} & \multicolumn{2}{c|}{Odd} & \multicolumn{2}{c|}{All} & \multicolumn{2}{c|}{Even} & \multicolumn{2}{c}{Odd} \\
\cmidrule(lr){2-13}
& w/ m & w/o m & w/ m & w/o m & w/ m & w/o m & w/ m & w/o m & w/ m & w/o m & w/ m & w/o m \\
\hline \addlinespace
\makecell[l]{\cb{}} & $\mathbf{62.3}$\% & $61.8$\% & $\mathbf{64.09}$\% & $63.7$\% & $\mathbf{64.68}$\% & $64.09$\% & $\mathbf{90.57}$\% & $89.63$\% & $\mathbf{90.39}$\% & $90.30$\% & $\mathbf{92.17}$\% & $91.06$\% \\ \addlinespace
\makecell[l]{\gcb{}} & $\mathbf{61.49}$\% & $61.05$\% & $\mathbf{64.02}$\% & $62.2$\% & $\mathbf{63.79}$\% & $62.66$\% & $\mathbf{90.74}$\% & $90.71$\% & $\mathbf{92.13}$\% & $91.63$\% & $\mathbf{91.46}$\% & $91.06$\% \\ \addlinespace
\makecell[l]{\unx{}} & $\mathbf{65.44}$\% & $64.20$\% & $\mathbf{65.62}$\% & $64.49$\% & $\mathbf{65.59}$\% & $65.52$\% & $\mathbf{94.7}$\% & $93.64$\% & $\mathbf{94.5}$\% & $93.62$\% & $\mathbf{95.47}$\% & $94.4$\% \\
\bottomrule
\end{tabular}
}
\label{tab:tok_merging}
\end{table}
During pruning, the representation of tokens marked to be pruned can be entirely removed or merged into one vector. To study the impact of token merging, we conduct experiments to measure its impact on the two tasks and report the results in Table~\ref{tab:tok_merging}.

Across all models and datasets, merging tokens consistently yields better performance compared to discarding them completely. \unx{} benefits the most from this design choice, with an average performance improvement of $1.14$\% on \devign{} and $1$\% on \bcb{}. 

On the \devign{} dataset, \gcb{} shows the highest performance gain when pruning even layers ($1.82$\%), while \cb{}'s performance is enhanced the most when pruning odd layers on the BCB dataset ($1.11$\%). In contrast, \gcb{} exhibits the least improvement when pruning all layers on the \bcb{} dataset ($0.03$\%).

\smrybx{\textbf{Summary}: 
Merging prunable tokens into a single vector consistently outperforms compared to completely discarding them across all models and datasets, with an average performance improvement of $0.77$\%.
% Merging pruned tokens into a single vector consistently outperforms completely discarding them across all models and datasets, with an average performance improvement of $0.77$\%. 
However, the extent of such gain depends on the language model and the task.}

\subsection{\textbf{RQ4: Impact of Computing Environment}}
In Figures~\ref{fig:a100_speed} and~\ref{fig:rtx2080_speed}, we plot the time taken to fine-tune the language models on the selected tasks when using different \gpu{}s.
First, examining the results for the A$100$ \gpu{}, we observe that \alpine{} leads to a reduction in fine-tuning time for all three models on both the tasks. For defect prediction, pruning results in a running time reduction of around $25$\% for \cb{}, $24$\% for \gcb{}, and 15\% for \unx{}. Similarly, for the clone detection task, pruning reduces the running time by approximately $18$\% for \cb{}, $11$\% for \gcb{}, and $12$\% for \unx{}.  These findings demonstrate that pruning is effective in improving the efficiency of the models on a high-performance \gpu{}.

\begin{wrapfigure}{r}{0.55\columnwidth}
\raggedright
    \includegraphics[width=0.55\columnwidth]{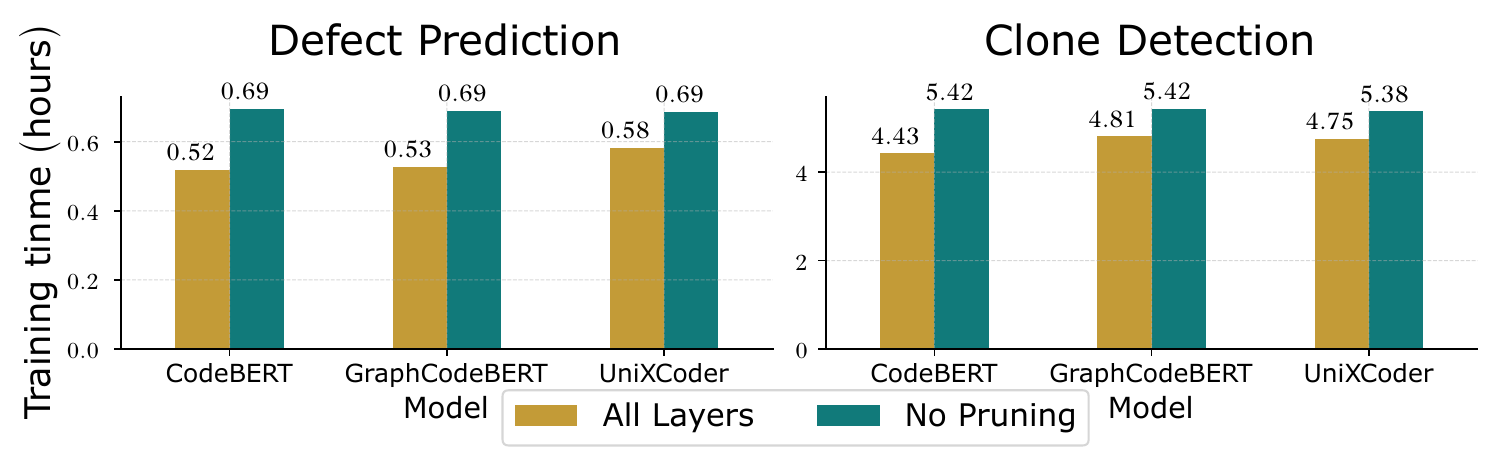}
    \caption{Fine-tuning time before and after pruning on an NVIDIA A100 \gpu{}.}
    \label{fig:a100_speed}
\raggedright
    \includegraphics[width=0.55\columnwidth]{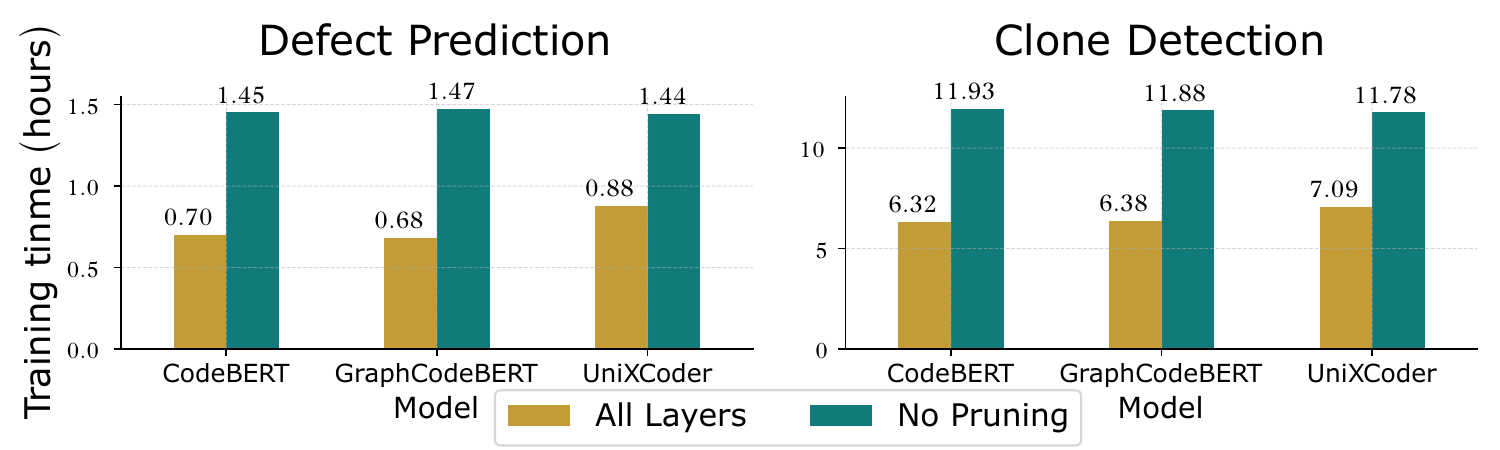}
    \caption{Fine-tuning time before and after pruning on an NVIDIA RTX2080 \gpu{}.}
    \label{fig:rtx2080_speed}
\end{wrapfigure}
Next, comparing the results between the RTX$2080$ and A$100$ \gpu{}s, we notice that the running times on the A$100$ are generally shorter than those on the RTX$2080$, both with and without pruning. This is expected due to the superior computational capabilities of the A$100$. However, the relative impact of pruning on efficiency differs between the two \gpu{}s. On the RTX$2080$, pruning leads to a more significant reduction in running time, with improvements ranging from $39$\% to $54$\% across models and tasks. In contrast, on the A$100$, the reduction in running time due to pruning is less pronounced, ranging from $11$\% to $25$\%.
\begin{wraptable}{l}{.5\columnwidth}
  \centering
  \caption{\gpu{} emission statistics across two NVIDIA \gpu{}s.}
  \resizebox{.5\columnwidth}{!}{
  \begin{tabular}{lccc}
    \toprule
    \gpu{} & \makecell[c]{E(CO\textsubscript{2} kg) \\ without \alpine{}} & \makecell[c]{E(CO\textsubscript{2} kg) \\ with \alpine{}} & \makecell[c]{Emission \\ reduction rate} \\
    \midrule
    RTX$2080$  & $5.93$ & $0.34$ & $44.85$\% \\
    A$100$      & $3.15$ & $2.69$ & $14.60$\% \\
    \bottomrule
  \end{tabular}
  }
\label{tab:co2_emissions}
\end{wraptable}

Building upon these results, we further investigate the implications of the reduced time from an environmental perspective.

In Table~\ref{tab:co2_emissions}, we report the CO\textsubscript{2} emission rates for both \gpu{}s with and without \alpine{}. The measurements were taken using the tool provided by Lacoste \etal{}~\cite{Lacoste2019Quantifying}.

The results show that \alpine{} significantly reduces the CO\textsubscript{2} emissions of both the RTX$2080$ and A$100$ \gpu{}s. For the RTX$2080$, the amount of CO\textsubscript{2} emitted decreases from $5.93$ kg to $0.34$ kg with \alpine{}, resulting in a substantial emission reduction rate of $44.85$\%. Similarly, for the A$100$, emissions decrease from $3.15$ kg to $2.69$ kg, resulting in a reduction rate of $14.60$\%. The reduction in 
%the amount of 
emissions is particularly significant for the RTX$2080$, which aligns with the previous analysis showing that pruning leads to a more substantial improvement in efficiency on consumer-grade \gpu{}s compared to high-performance \gpu{}s like the A$100$.

\smrybx{\textbf{Summary}: 
We observe significant improvements in the efficiency of models thanks to pruning, especially on consumer-grade \gpu{}s such as the RTX$2080$ compared to high-performance \gpu{}s like the A$100$,
% The impact of pruning on the efficiency of 
% \cb{}, \gcb{}, and \unx{} 
% models varies depending on the computing environment, with more significant improvements observed on consumer-grade \gpu{}s such as the RTX$2080$ compared to high-performance \gpu{}s like the A$100$,
highlighting the effectiveness of such method in enabling the use of these models on less powerful hardware.
Furthermore, the significant emission reduction rate observed on the RTX$2080$ underscores the potential of \alpine{} to enable the sustainable adoption of language models on a wider range of hardware, including consumer-grade \gpu{}s. 
% which are more readily available and cost-effective.
}
\section{Discussion}
\subsection{Sequence Compression}
In Section~\ref{results:rq1} we demonstrated how \alpine{} can significantly reduce sequences' lengths to $\sim\times3$ less than the original lengths in some scenarios.
In code analysis tasks, not all parts of the input sequence contribute equally to the final prediction. Many tokens in source code, such as common language keywords, special characters, or even standard library imports, may not be critical for identifying vulnerabilities or detecting code clones. However, as previously stated in Section~\ref{results:rq1}, the reduction rate depends on the nature of the task. 

Vulnerability detection, which is the focus of the \devign{} dataset, usually involves identifying specific code patterns that may lead to security issues. These vulnerability indicators often compromise a small portion of overall function, allowing for rapid pruning of irrelevant code segments. The exponential-like decay observed in Figure~\ref{fig:seq_len_progression} suggests that the models may quickly isolate potentially vulnerable code portions, discarding a large amount of non-critical code in earlier layers.

On the other hand, in code clone detection, the task associated with \bcb{}, requires a more holistic view of code structure and functionality. Clones can manifest in various forms, from exact verbatim copies (\ie{} Type-1 clones) to functionality-similar but syntactically different implementations, also known as Type-4 clones. This entails the need to retain a larger set of code features through the layers, which explains the more gradual, linear-like reduction in sequence length. The models likely need to maintain a wider context to effectively compare and identify similarities across different code snippets.
\subsection{Less is More}
Our findings on sequence reduction with \alpine{} caters to the calls of developing smaller, task-oriented models~\cite{Menzies2023}. Just as compact models demonstrate that not all parameters in large, general-purpose models are necessary for specific tasks, our results show that not all tokens in an input sequence are crucial for effective code analysis tasks. In both cases, the underlying principle is the same: \textit{task-specific efficiency}. (Large) language models, while powerful, often contain redundant or task-irrelevant information. Similarly, full code sequences include many tokens that may not contribute significantly to solving SE tasks. This is also similar to previous research. For instance, Xu \etal{}~\cite{Xu2023} have observed that removing noisy instances when training models for obsolete comment detection can lead to up to $10.7$\% improvement in accuracy, with $15$\% less data.
The findings of this work align with the principle of \textit{less is more}. We can achieve efficient, high-performing systems not by processing more data, but by intelligently selecting the most pertinent information and avoiding noisy data for the task at hand.

\subsection{Extension of \alpine{} to other Transformer Variants}
\label{discussion:extention}
The main focus of \alpine{} was encoder-only models given their ability to generate embeddings that are primarily suitable for identification tasks. They are also an important component in text generation systems that use information retrieval, or Retrieval Augmented Generation (RAG). However, it is important to remark that \alpine{} can be extended to other variants of the Transformer architecture, \ie{} decoder-only and encoder-decoder models. In decoder-only models, \alpine{} would be placed after the Masked Self-attention layer ({\sc MaskedSA}), not after the self-attention layer because the {\sc MaskedSA} attends only to previous tokens, preserving the autoregressive nature of the decoder. In addition, the last token would be kept, as currently done, so that it can be used by LM head. Finally, in the encoder-decoder setting, \alpine{} can be either applied to the encoder, decoder, or both. The evaluation of \alpine{}'s performance on these architectures is deferred to future work. In return, this would increase their adoption given the less computational resources they demand, especially on consumer-grade hardware.
% Discussion points:
% - Impact of energy consumption.
% - High compression ratio of sequences while maintaining high performance. Evidence of noisy data.
% - How can this be extended to other variants of the transformers architecture.
% - ALPINE and other orthogonal compressions techniques.
\section{Threats to Validity}
\noindent
\textit{\textbf{Internal validity}}:
Internal threats to validity are concerned with the ability to draw conclusions from our experimental results.
The hyperparameters specified for a model influence the time, resources required, and the results.
To mitigate a potential threat, we used the same hyperparameters across all models, as reported in previous works~\cite{Lu2021CodeXGLUE}, and utilized the same data splits for fine-tuning and testing.
% By keeping the same hyperparameters, we aimed to control for potential confounding variables that could influence the results. 
This consistency in the experimental setup strengthens the internal validity of our findings, ensuring that the observed effects can be attributed to the \alpine{} pruning method rather than variations in hyperparameters or data splits.
% \smallskip
\\
\textit{\textbf{External validity}}:
Such threats are concerned with the ability to generalize our results. 
The language-agnostic nature of the \alpine{} pruning method allows it to be applied to a wide range of programming languages without the need for language-specific adaptations.
We evaluate the approach on three language models and tasks with different programming languages (C and Java).
Its compatibility with all transformer-based encoder models makes it trivial to integrate into various architectures, enhancing its generalization.
% Furthermore, its compatibility with all Transformer encoder-based models enhances its generalizability, as it can be readily integrated into various model architectures.
Finally, the two \gpu{}s used in this study, NVIDIA A$100$ and RTX$2080$, belong to \textit{Ampere} and \textit{Turning} architecture families. These architectures encompass a larger set of other cards which we expect to exhibit the similar reported trends.
% \smallskip
\\
\textit{\textbf{Construct validity}}: 
Construct threats to validity are concerned with the degree to which our analyses measure what we claim to analyze.
We employed well-established metrics for measuring predictive performance (such as, accuracy and F1-score), 
computational efficiency (\flops{}, throughput, running time and memory footprint),
and environmental impact through CO\textsubscript{2} emissions rate.
Relying on the traditional metrics for the measured aspects mitigates the potential threats of construct validity and ensures that the variables are accurately captured.
\vspace{-2mm}
\section{Related Work}

\subsection{Pre-trained language models of code}

Pre-trained code language models have demonstrated benefits in various SE tasks, including code search, code and documentation generation, and defect prediction~\cite{Wang2021codet5,Guo2022,Lu2021CodeXGLUE,Sharma2024}. Typically based on the transformer architecture~\cite{Vaswani2017}, these models can broadly be categorized into decoder-only, encoder-only, and encoder-decoder language models~\cite{min2023recent}.

\smallskip
\noindent\textbf{Decoder-only architectures} are typically pre-trained in an autoregressive manner, \ie{} predicting the subsequent token based on preceding ones. In the realm of code, examples include \textsc{InCoder}~\cite{Fried2023}, \textsc{StarCoder 2}~\cite{Lozhkov2024Starcoder}, and \textsc{CodeGen 2}~\cite{nijkamp2023codegen2}. These architectures excel in code generation tasks.

\noindent\textbf{Encoder-only architectures} are primarily pre-trained on the masked language modeling task \ie{} these models are trained to predict masked words given the remaining sequence as context. Examples include \cb{}~\cite{Feng2020} and \gcb{}~\cite{Guo2021}. They are frequently used for classification tasks or as a means of feature extraction.

\noindent\textbf{Encoder-decoder architectures} undergo pre-training through conditional generation tasks \ie{} generating text based on provided input text, \eg{} for code summarization or generating code from natural language descriptions. Sample architectures include {\sc CodeT5}~\cite{Wang2021codet5}, {\sc CodeT5}+\cite{wang2023codet5+}, and \textsc{plbart}~\cite{ahmad2021unified}.
\vspace{-2mm}
\subsection{Optimized transformers}
The techniques to enhance the effectiveness of a transformer model while maintaining its performance fall into three categories---knowledge distillation, quantization, and pruning~\cite{Shi2023,Kim2022LTP,zhu2023survey}.
% There are several methods to enhance the effectiveness of a given transformer model while maintaining its performance. These techniques fall into three categories: knowledge distillation, quantization, and pruning~\cite{Shi2023,Kim2022LTP,zhu2023survey}.

\smallskip
\noindent\textbf{Knowledge distillation.} This technique trains a smaller model (the \textit{student} model) to imitate the behavior of a larger model (the \textit{teacher} model)~\cite{gou2021knowledge}. In a SE context, Shi \etal{}~\cite{Shi2023} introduced \textit{Compressor}, a framework that employs task-specific knowledge distillation to improve the efficiency of the final transformer model. The authors perform a neural architectural search via a genetic algorithm, followed by training the selected model using knowledge distillation. In their experiments, they use \cb{}~\cite{Feng2020} and \gcb{}~\cite{Guo2021} as teacher models to tackle the tasks of vulnerability prediction and code clone detection. \textit{Compressor} still requires the fine-tuning of the teacher model, which incurs overhead throughout the process. \alpine{} allows to reduce the computational costs of fine-tuning these models, minimizing such an overhead.

\noindent\textbf{Quantization.} It converts the model's parameters from 32-bit floating-point numbers (\texttt{float32} data type) to lower-precision formats, such as 8-bit integers (\texttt{int8} data type)~\cite{nagel2021white}. 
Depending on the component of quantization, these techniques can be categorized as quantization of weights only \cite{xu2023qa,dettmers2024qlora,kim2023squeezellm} and quantization of both weights and activations \cite{zafrir2019q8bert,yao2022zeroquant,dettmers2023spqr}. Although quantization reduces the memory footprint of these models, it may also reduce the inference speed and accuracy~\cite{jin2024comprehensive}.
% For example, Zafrir \etal{}~\cite{zafrir2019q8bert} quantize the network parameters and activations of BERT to 8-bit integers during fine-tuning. 
% To harness the advantages of quantization, specialized hardware supporting low-bit operations or optimized low-level processing libraries are necessary during inference~\cite{xx,xx,xx}. \todo{Idk if this is last sentence is correct.}

% \smallskip
\noindent\textbf{Pruning.} Pruning techniques can be divided into \textit{weight pruning} and \textit{token pruning}~\cite{Kim2022LTP}. 
Liu \etal{}\cite{liu2021ebert} adopt weight pruning to prune unimportant heads and channels from the \mha{} and \ffnn{} respectively. 
Zang \etal{}~\cite{Zhang2022DietCode} offer DietCode---a token pruning approach.
To build DietCode, Zang \etal{} carry out an empirical investigation to determine the significance of statements and tokens for \cb{}. Subsequently, based on these findings, they employ a 0-1 Knapsack approach to prune the input sequence. In contrast, our approach is performed adaptively at each layer. Additionally, during the pruning step, we could aggregate the pruned tokens into a single one to prevent excessive loss of information. Moreover, our method is language-agnostic, facilitating easy adaptation to diverse programming languages and extending applicability to natural language inputs, particularly beneficial in minimizing computational overhead in tasks like code search or code generation from documentation.

As for LTP~\cite{Kim2022LTP}, in addition to the main limitations previously mentioned where concrete computational speed-up is not achieved, there exist some limitations concerning the methodology. First, the reported number of FLOPs is manually calculated, under the assumption that matrix multiplication is optimized when zeroed-out vectors are present. Furthermore, given the placement of the pruning layer, these representations do not remain zeroed-out given the presence of skip connections and normalization layers\footnote{This was also pointed out by the issues raised in the replication package repository: \url{https://t.ly/zBMGE} and \url{https://t.ly/kLCiU}.}. Moreover, it does \emph{not} mask tokens during the evaluation,\ie{} the full sequence, threatening the validity of the accuracy results.
\vspace{-2mm}
\section{Conclusions and Future Work}
This study addresses the challenges,
specifically computational efficiency and environmental impact,
associated with language models of code. 
We propose \alpine{}, a pruning method that reduces the input sequence while maintaining predictive performance. 
Our experiments on three language models and datasets with different programming languages
demonstrate that it significantly improves computational efficiency that, in turn, reduces CO\textsubscript{2} emissions.
% for Transformer-based models.
The results show that it is particularly effective on consumer-grade \gpu{}s, enabling the usage of these models on more accessible hardware.
Furthermore, the programming language-agnostic nature of \alpine{} and its compatibility with Transformer-based models enhance the generalizability of our findings. 
We envision to extend this work along two axes. 
First, to assess \alpine{} on encoder-decoder, and decoder-only models. Second, we aim to investigate other scoring importance functions and more adaptable merging strategies. 

\section{Data Availability}
Our replication package including source code and data is available online through this repository~\cite{replication}.
\begin{acks}
This work is partially funded by by the Natural Sciences and Engineering Research Council of Canada (NSERC) through grant NSERC Discovery RGPIN/$04903$.
\end{acks}

\bibliographystyle{ACM-Reference-Format}
\bibliography{refs}

\end{document}